\documentclass[final,5p,times,twocolumn]{elsarticle}

\usepackage{graphicx}      % include this line if your document contains figures
\usepackage{natbib}        % required for bibliography
\usepackage{amsmath}  
\usepackage{amssymb}
\usepackage{mathtools}
\usepackage{booktabs}
\usepackage{url}
\usepackage{amsthm}

\newtheorem{hypo}{Hypothesis}

\newtheorem{remark}{Remark}
\newtheorem{proposition}{Proposition}
\newtheorem{lem}{Lemma}
\newtheorem{thm}{Theorem}
\usepackage[subtle]{savetrees}
\usepackage{xcolor}
\begin{document}
\begin{frontmatter}

\title{Adaptive Sliding Mode Formation Control for Space Interferometer Missions} 

%\thanks[footnoteinfo]{Sponsor and financial support acknowledgment goes here. Paper titles should be written in uppercase and lowercase letters, not all uppercase.}
\author[PoliTo]{Mauro Mancini}\ead{mauro.mancini@polito.it}
\author[PoliTo]{Giulia Alessandra Tataru}\ead{s322755@studenti.polito.it}
\author[Osaka]{Satoshi Satoh}\ead{satoh@mech.eng.osaka-u.ac.jp)}
\author[PoliTo]{Elisa Capello\corref{cor1}}\ead{elisa.capello@polito.it}

\affiliation[PoliTo]{organization={Department  of Mechanical and Aerospace Engineering, Politecnico di Torino}, 
addressline={Corso Duca degli Abruzzi 24},  
city={Torino}, 
postcode={10129}, 
state={}, 
country={Italy}} 
\affiliation[Osaka]{organization={Department of Mechanical Engineering, The University of Osaka}, 
addressline={2-1 Yamadaoka, Suita},  
city={Osaka}, 
postcode={565-0871}, 
state={}, 
country={Japan}} 

\cortext[cor1]{Corresponding author} \fntext[]{Received 31 October 2025, Revised 6 February 2026, Accepted 20 April 2026 \newline Digital Object Identifier (DOI): \url{10.1016/j.conengprac.2026.107025}}

% \address[PoliTo]{Department of Mechanical and Aerospace Engineering, Politecnico di Torino, Corso Duca degli Abruzzi 24, 10129 Torino, Italy
% (e-mails: s322755@studenti.polito.it, mauro.mancini@polito.it, elisa.capello@polito.it)}
% \address[Osaka]{Department of Mechanical Engineering, The University of Osaka, 565-0871, Japan (e-mail: satoh@mech.eng.osaka-u.ac.jp)}

\begin{abstract}                % Abstract of 50--100 words
This paper addresses high-precision formation control for spacecraft operating in low Earth orbit, motivated by the requirements of future space interferometry missions such as SILVIA. The proposed approach formulates the relative dynamics within a port-Hamiltonian framework and introduces an Adaptive Boundary-layer Sliding Mode Control (AB-SMC) law to overcome the limitations of conventional SMC with constant gains. The key innovation lies in a dynamic, error-dependent adjustment of the sliding manifold, enhancing transient performance while guaranteeing high-precision trajectory tracking.  Rigorous Lyapunov-based analysis establishes explicit ultimate bounds on the tracking error and ensures closed-loop stability, while extensive Monte Carlo simulations further validate the proposed AB-SMC compared to standard control approaches. Results show that AB-SMC achieves faster convergence, lower control effort, and sub-millimeter tracking accuracy, demonstrating its practical robustness and implementation feasibility in realistic, uncertain orbital environments while respecting low-thrust constraints.
\end{abstract}


\end{frontmatter}
%===============================================================================

\textcopyright 2026. This manuscript version is made available under the CC BY 4.0 license \url{https://creativecommons.org/licenses/by/4.0/}

\section{Introduction}
In recent years, spacecraft formation flying has emerged as a critical technology for advancing space science and exploration. 
Distributed spacecraft architectures present several advantages over traditional monolithic platforms, including enhanced reliability, robustness, improved efficiency, and cost-effectiveness \citep{doi:10.2514/1.G002868}. Moreover, formation flying enables mission types that a single spacecraft cannot undertake, such as high-precision interferometry, distributed sensing, and three-dimensional planetary surface reconstruction \citep{schweighart2001satellite}. Recent research has also shown that vision-based control can enable efficient formation keeping of spacecraft \citep{Pomares2024}.
The growing interest in this technology is evidenced by several current and planned missions, such as the Laser Interferometer Space Antenna (LISA) \citep{DR03, XJL24z}, the Deci-hertz Interferometer Gravitational Wave Observatory (DECIGO) \citep{KawamuraEtc08}, the Large Interferometer for Exoplanets (LIFE) \citep{QuanzEtc22}, the Space Experiment of Infrared Interferometric Observation Satellite (SEIRIOS) \citep{IkariEtc21, MIKINY22}, and the forthcoming Space Interferometer Laboratory Voyaging towards Innovative Applications (SILVIA) mission \citep{SILVIA25}. 
%added by SS
\textcolor{black}{Notably, in-orbit verification of precision formation flying is ongoing. ESA's Proba-3, launched in December 2024, is now entering its flight phase, aiming for precision formation flying technology demonstration and solar coronagraphy. Using an integrated navigation architecture combining relative GPS, vision-based sensors, and laser ranging, the two independent $300$ kg spacecraft have achieved millimeter-level relative position control and arcsecond-level line-of-sight control in a highly elliptical orbit at a separation of approximately $150$ meters \citep{SerranoEtc25}. 
More recently, small-satellite-based technology demonstration missions for precision formation flight have been actively pursued, including CNES’s POC\_ESSAIM \citep{MDDL24} and the VISORS mission developed through multi-institution collaboration \citep{KDL23}.}
Additional examples of formation-flying missions can be found in \citep{aerospace9070362}. 

Despite these potential benefits, formation-flying missions demand unprecedented levels of precision in maintaining the relative positions and orientations among multiple spacecraft \citep{doi:10.2514/1.G002868}. The primary challenge lies in the design of robust and reliable Guidance, Navigation, and Control (GNC) systems capable of operating autonomously in the orbital environment. In particular, the control function must ensure high-accuracy formation maintenance under complex, highly nonlinear dynamics and in the presence of external perturbations such as the $J_2$ effect and atmospheric drag. Furthermore, the control laws must produce physically realizable actuation commands, maintaining low control effort to enable long-duration missions while preventing actuator saturation, a critical constraint for spacecraft equipped with low-thrust propulsion systems.
%added by SS
\textcolor{black}{
Although formation control has been extensively studied for a wide range of terrestrial platforms, including unmanned aerial, ground, and underwater vehicles \citep{BAW11,Ahn20}, the aforementioned aspects constitute key characteristics of spacecraft formation control. In spacecraft applications, control strategies must actively exploit natural orbital dynamics to maintain fuel-efficient configurations; therefore, position-based control approaches are often preferred, as they enable the use of pre-designed reference orbits that satisfy both mission requirements and orbital mechanics. Accordingly, this paper adopts a formation tracking control approach.
\subsection{Mission Context}}
This paper focuses on the problem of precise trajectory tracking within the framework of the SILVIA mission mentioned above, whose primary objective is to demonstrate precision formation-flying technologies in low Earth orbit (LEO). The selected orbital configuration consists of a circular orbit of 500 km altitude, where three spacecraft maintain a triangular formation with 100 m separation while orbiting along a General Circular Orbit (GCO) \citep{SILVIA25}.
The LEO environment subjects the formation to significant perturbations, mainly atmospheric drag and gravity field irregularities, that greatly complicate precise relative motion control. As SILVIA mission specifies a submillimeter level relative positioning requirement, achieving such accuracy under these conditions represents a demanding challenge, emphasizing the need for robust and reliable control algorithms.

Following the approach in \citep{JYY20, SH26, IchitoTABUCHI2024105}, this work treats the nonlinear relative motion equations in formation flying as a class of port-Hamiltonian systems \citep{MS92,DMSB09}, without relying on linear approximations. This approach captures the intrinsic nonlinear dynamics of the formation, while also accounting for external perturbations, such as $J_2$ effects and atmospheric drag. %Moreover, an error system is defined through a generalized canonical transformation, preserving the port-Hamiltonian structure. 
In this framework, the tracking to the target trajectory in the original system and  the convergence to the origin in the error system are in one-to-one correspondence, thus enabling systematic stability analysis.
%In \citep{IchitoTABUCHI2024105}, the properties of port-Hamiltonian structure was exploited to prove the stability of first-order trajectory-tracking Sliding Mode Control (SMC).

\textcolor{black}{Several nonlinear control approaches could be considered for the addressed problem, including adaptive backstepping controller %for spacecraft attitude maneuvers \citep{Ren2025},
combined with Sliding Mode Control (SMC) \citep{https://doi.org/10.1155/2024/6847067}, 
%Nonlinear Model Predictive Control (NMPC) for autonomous spacecraft attitude guidance \citep{Relvas2025, Kondo2024}, 
distributed Model Predictive Control (MPC) with leader-follower configuration \citep{9964448}, time-varying observer-based control strategies \citep{Shao01112024}, 
adaptive control methods \citep{9793721}, control schemes based on the State-Dependent Differential Riccati Equation (SDDRE) \citep{aerospace12030201}, optimal control \citep{10609539}, and learning-based methods \citep{app12052485}.
%feedback linearization techniques in advanced spacecraft control studies \citep{IAC2025}. 
These methods have been investigated to address several issues in spacecraft formation flight, such as fuel consumption, actuator faults, constrained maneuvering, and trajectory optimization in complex dynamical environments. In this work, Sliding Mode Control (SMC) was selected due to its relatively simple implementation, strong robustness with respect to model uncertainties and external disturbances. Compared to more computationally demanding approaches such as nonlinear MPC, SMC provides a favorable trade-off between performance and implementation complexity, making it particularly suitable for mission-grade onboard applications.}
Starting from \citep{IchitoTABUCHI2024105}, a first-order trajectory-tracking Sliding Mode Control (SMC) is proposed. SMC is a nonlinear control technique known for its precision and robustness \citep{utkin2013sliding}, whose design involves selecting a sliding surface with attractive convergence properties and defining a discontinuous control law to confine the system trajectories to this surface. In particular, SMC can handle model uncertainties and external perturbations by driving the system states onto a so-called sliding surface, thereby achieving desirable dynamics once on that surface. 
To mitigate the well-known issue of chattering (high-frequency oscillations in the closed-loop dynamics), the authors of \citep{mancini25} introduced a Boundary layer SMC (B-SMC) \citep{slotine1983tracking}, smoothing the discontinuity of the control action near the sliding surface. The B-SMC alleviates the deleterious effects of chattering — which in formation flight could mean frequent small rapid actuator changes, inter-vehicle jitter, or communications stress — by replacing the sign-function discontinuity with a saturation or continuous approximation within a boundary thickness. This smoother behavior improves the use of actuators and mitigates potential excitation of high-frequency dynamics, which is particularly important when coordinating many vehicles in close proximity. 
\textcolor{black}{It is well recognized that the introduction of a boundary layer in SMC weakens the convergence properties of the ideal discontinuous formulation, resulting in ultimate boundedness of the tracking error. %Nevertheless, the tuning strategy adopted in this work allows the bound on the tracking error to be explicitly selected as a design parameter, thus guaranteeing uniform boundedness with a prescribed level of tracking accuracy. 
%Although continuous Higher-Order SMC (HOSMC) effectively attenuate chattering while preserving ideal convergence properties in continuous time, their closed-loop performance are particularly sensitive to discretization effects and reduced sampling rates when implemented on digital platforms \citep{UTKIN202010244, Boiko2009, shtessel2014sliding}. Given the limited computational resources and sampling frequencies typically available on onboard spacecraft computers, B-SMC represents a practically reliable baseline for the problem under consideration.}
A viable strategy for reducing chattering while preserving ideal convergence properties is provided by continuous Higher-Order SMC (HOSMC). However, these controllers are particularly sensitive to discretization effects and reduced sampling rates when implemented on digital platforms \citep{UTKIN202010244, Boiko2009, shtessel2014sliding}. Considering the limited computational resources typically available on spacecraft onboard computers, B-SMC represents a more reliable reference basis for the problem under consideration compared to HOSMC.}

\textcolor{black}{\subsection{Proposed Control Framework and Main Contributions}
In this work, we first establish a Lyapunov-based proof of closed-loop stability under B-SMC, providing an explicit parametrization of the resulting invariant set as a function of the control gains. This analysis reveals a fundamental limitation of constant-gain B-SMC, namely the inherent trade-off between transient performance and steady-state tracking accuracy.
Specifically, the choice of fixed gains imposes a compromise between smooth control action and precise steady-state tracking, which cannot be alleviated through static gain tuning alone. To overcome this limitation, we introduce a novel Adaptive Boundary-layer Sliding Mode Control (AB-SMC) law, in which the parameters of the sliding manifold are dynamically adjusted based on the system states.}
%In this work, we first provide a Lyapunov-based proof of closed-loop stability under B-SMC, including a quantitative parametrization of the invariant set in terms of the control gains. This analysis highlights a fundamental limitation of constant-gain B-SMC, which requires a trade-off between transient performance and steady-state error.

\textcolor{black}{Although adaptive and boundary-layer SMC have been widely investigated, the sliding manifold is typically assumed to be fixed and designed a priori.
In classical adaptive SMC, the control gain is increased to force sliding motion even under unknown uncertainties and disturbances \citep{Plestan01092010}. Although robust, this paradigm inherently promotes fast and aggressive closed-loop dynamics, making it unsuitable for systems with limited control authority.
Within the boundary-layer SMC framework, time-varying and state-dependent boundary-layer widths have been proposed to mitigate the trade-off between chattering attenuation and tracking accuracy \citep{1039802}. However, transient behaviour remains dictated by the fixed geometry of the sliding manifold and can still cause excessively aggressive dynamics in input-constrained systems.
%In contrast, the proposed AB-SMC adopts a different paradigm: the control gain and the boundary-layer thickness are kept constant, while a purely state-dependent adaptive law is used to shape the sliding manifold parameters. %This approach directly regulates the closed-loop error dynamics, enabling smooth, low-effort transients for large tracking errors and high-precision steady-state behavior as the system approaches the desired configuration, without increasing control aggressiveness.
}
\textcolor{black}{In contrast, the proposed AB-SMC adopts a different paradigm: the control gain and the boundary-layer thickness are kept constant, while a purely state-dependent adaptive law is used to shape the sliding manifold parameters. 
The idea of adapting the parameters of the sliding manifold has already been explored in the literature. Several works rely on time-scheduled laws, in which the manifold parameters evolve according to predefined time profiles \citep{tokat2015classification}. In these approaches, the closed-loop behavior is inherently tied to a nominal transient, which limits the ability of the controller to respond effectively to changes in the reference trajectory or operating conditions.
Other contributions introduce state-dependent adaptation through fuzzy-logic mechanisms \citep{tokat2003}, at the price of increased design complexity and computational burden. In contrast, the proposed AB-SMC employs a simple state-dependent adaptive law based on a low-complexity differential equation, which directly shapes the sliding manifold while avoiding the need for defuzzification process.
The proposed AB-SMC framework is rigorously analyzed via Lyapunov methods and validated through numerical simulations, demonstrating smooth transients for large tracking errors and high-precision steady-state behavior as the system approaches the desired configuration, without increasing control aggressiveness. %improved transient response, reduced control effort, and high-precision trajectory tracking, while preserving closed-loop stability. 
%To overcome this limitation, we introduce a novel Adaptive Boundary-layer SMC (AB-SMC) law and rigorously characterize its performance through Lyapunov-based analysis, demonstrating improved transient response, reduced control effort, and high-precision trajectory tracking without compromising dynamic stability.
}

%In this work, we first provide a Lyapunov-based proof of closed-loop stability under B-SMC, including a quantitative parametrization of the invariant set in terms of the control gains. This analysis highlights a fundamental limitation of constant-gain B-SMC, which requires a trade-off between transient performance and steady-state error. To overcome this limitation, we introduce a novel Adaptive Boundary-layer SMC (AB-SMC) law and rigorously characterize its performance through Lyapunov-based analysis, demonstrating improved transient response, reduced control effort, and high-precision trajectory tracking without compromising dynamic stability.
%To assess the effectiveness and reliability of the proposed control strategies, an extensive numerical simulation campaign is conducted under realistic orbital conditions. The study systematically evaluates four different control approaches: Passivity-Based Control (PBC), SMC, constant gains B-SMC, and AB-SMC with time-varying gains. 
The numerical simulations campaign is conducted under reailistic orbital conditions, and the study systematically evaluates four different control approaches: Passivity-Based Control (PBC), SMC, constant gains B-SMC, and AB-SMC with time-varying gains.
The simulations are designed to rigorously assess the controller performance in the presence of orbital perturbations, such as the effect of $J_2$ and atmospheric drag, as well as across a wide range of initial conditions. Particular attention is given to the ability of each controller to maintain precise relative positioning, to mitigate the impact of external disturbances over extended mission durations, and to minimize control effort, which is critical for practical implementation with low-thrust actuators. The simulation framework not only serves to validate the theoretical control designs but also provides essential insights into their practical applicability and efficiency in operational space environments. By systematically testing all control schemes under external disturbances and various initial conditions, the study ensures a comprehensive assessment of robustness, stability, and control efficiency. The simulations demonstrate the feasibility of achieving high-precision spacecraft formations in GCO and effectively bridge the gap between theoretical nonlinear control design and its practical implementation in realistic orbital scenarios. 
The results highlight the inherent trade-offs between tracking accuracy, transient performance, and control effort across different control strategies. In particular, the proposed AB-SMC consistently outperforms the other approaches, achieving faster convergence, reduced control effort, and superior trajectory-tracking precision, demonstrating its potential as a highly effective solution for next-generation formation-flying missions.

The rest of the paper is organized as follows. Section 2 presents the relative motion dynamics of the formation, formulated within the port-Hamiltonian framework. Section 3 details the control strategies, %including classical SMC, constant gains B-SMC, and the proposed AB-SMC,
along with the associated Lyapunov-based stability analyses. Section 4 reports the results of the numerical simulation campaign.%, evaluating controller performance in realistic orbital conditions and highlighting the advantages of AB-SMC. 
Finally, Section 5 draws the conclusions and outlines perspectives for future formation-flying missions.

\section{Spacecraft Relative Dynamics}
%\subsection{Port-Hamiltonian Representation of the Nonlinear Relative Motion}
The nonlinear relative dynamics between a deputy spacecraft and a chief spacecraft under the influence of the Earth's $J_2$ perturbation can be represented within the port-Hamiltonian framework, as stated in \citep{IchitoTABUCHI2024105}.\\
Let $Q \in \mathbb{R}^3$ denote the deputy’s position in the Earth-Centered Inertial (ECI) frame and $u_{\mathrm{ECI}} \in \mathbb{R}^3$ the control input. 
The deputy's dynamics is given by
\begin{equation}
    \ddot{Q} = -\frac{ \mu}{\|Q\|^3}Q + f_{J_2}(Q) + u_{\mathrm{ECI}},
    \label{eq:ECI_dynamics}
\end{equation}
where $\mu$ is the Earth’s gravitational parameter and $f_{J_2}(Q)$ represents the perturbation caused by the Earth's oblateness.\\
To describe the relative motion with respect to the chief, the ECI frame is transformed to the local-vertical–local-horizontal (LVLH) frame using the rotation matrix $T(t)\in SO(3)$.
%, whose columns are constructed from the chief’s position $Q_c$ and angular momentum $h_c = Q_c \times \dot{Q}_c$ 
The LVLH frame can be expressed as a rotational sequence given by $(\Omega_c, i_c, \theta_c)$ as the (3-1-3) Euler angles from the ECI frame, where $\Omega_c, i_c$ and $\theta_c$ denote the right
ascension of the ascending node, inclination, and argument of latitude of the chief, respectively
(see Fig.~\ref{fig:lvlh} for a schematic of the ECI and LVLH frames). The transformation satisfies $\dot{T}(t) = T(t)\,\omega_c^{\times}$, where $\omega_c \in \mathbb{R}^3$ is the angular velocity of the LVLH frame relative to the ECI frame and $(\cdot)^{\times}$ denotes the skew-symmetric operator.\\
Denoting by $q \in \mathbb{R}^3$ the deputy’s position with respect to the chief in the LVLH frame, the relative motion is described by the following equations, where $u=T(t)^{\top}u_{\mathrm{ECI}}$ is the control input expressed in the LVLH frame.
\begin{align*}
    \ddot{q} &= -\ddot{q}_c-(T(t)^{\top}\dot{T}(t)\,\omega_c^{\times}+\dot{\omega}_c^{\times})(q+q_c) 
    - 2\omega_c^{\times}(\dot{q}+\dot{q}_c) \nonumber \\
    & -\frac{\mu}{\|T(t)(q+q_c)\|^3}(q+q_c)
    + T(t)^{\top}f_{J_2}(T(t)(q+q_c))+u,
    \label{eq. LVLH_dynamics}
\end{align*}
Then, the Hamiltonian is derived as follows.
The gravitational potential including the $J_2$ perturbation is defined as
\begin{equation*}
    V(T(t)(q+q_c))=-\frac{\mu}{r}\left\{1 - \frac{J_2}{2}\left(\frac{R_e}{r}\right)^2\!\!\left(3\frac{Z^2}{r^2} - 1\right)\right\},
\end{equation*}
where $r=\|q+q_c\|$ and $R_e$ is the radius of the Earth.
Using the kinetic and potential energy, the Lagrangian is given by
\begin{equation*}\begin{split}
    L =& \frac{1}{2}(\dot{q} + \dot{q}_c + \omega_c^{\times}(q + q_c))^{\top}(\dot{q} +\dot{q}_c + \omega_c^{\times}(q + q_c)) \\
    &- V(T(t)(q + q_c)).
\end{split}\end{equation*}
The conjugate momentum is $p=\partial L^{\top} / \partial \dot{q}=\dot{q}+\dot{q}_c+\omega_c^{\times}(q+q_c)$, leading to the Hamiltonian
\begin{equation*}
    H(q,p,t)=\frac{1}{2}p^{\top}p - p^{\top}\!\bigl(\omega_c^{\times}(q + q_c) + \dot{q}_c\bigr) + V(T(t)(q + q_c)).
    \label{eq. Hamiltonian}
\end{equation*}

Let $x=[q^{\top}\; p^{\top}]^{\top}$ and $I_3$ denotes the 3$\times$3 identity matrix. Then, the nonlinear relative motion can be written in the port-Hamiltonian form:
% \begin{equation}
% \begin{cases}
% \dot{x}= \begin{bmatrix}
% O_{3,3} \hspace{5mm} I_3\\
% -I_3 \hspace{5mm} O_{3,3}
% \end{bmatrix}\frac{\partial H(x,t)^\mathsf{T}}{\partial x}+ \begin{bmatrix}
% O_{3,3} \\
% I_3 
% \end{bmatrix}u \vspace{6pt} \\ 
  
% y=[O_{3,3} \hspace{5mm} I_3]\frac{\partial H(x,t)^\mathsf{T}}{\partial x}
% \end{cases}
% \end{equation}
\begin{equation*}
\dot{x}= \begin{bmatrix}
O_{3,3} \hspace{5mm} I_3\\
-I_3 \hspace{5mm} O_{3,3}
\end{bmatrix}\frac{\partial H(x,t)^\mathsf{T}}{\partial x}+ \begin{bmatrix}
O_{3,3} \\
I_3 
\end{bmatrix}u 
\end{equation*}
\begin{figure}
    \centering
    \includegraphics[width=0.6\linewidth]{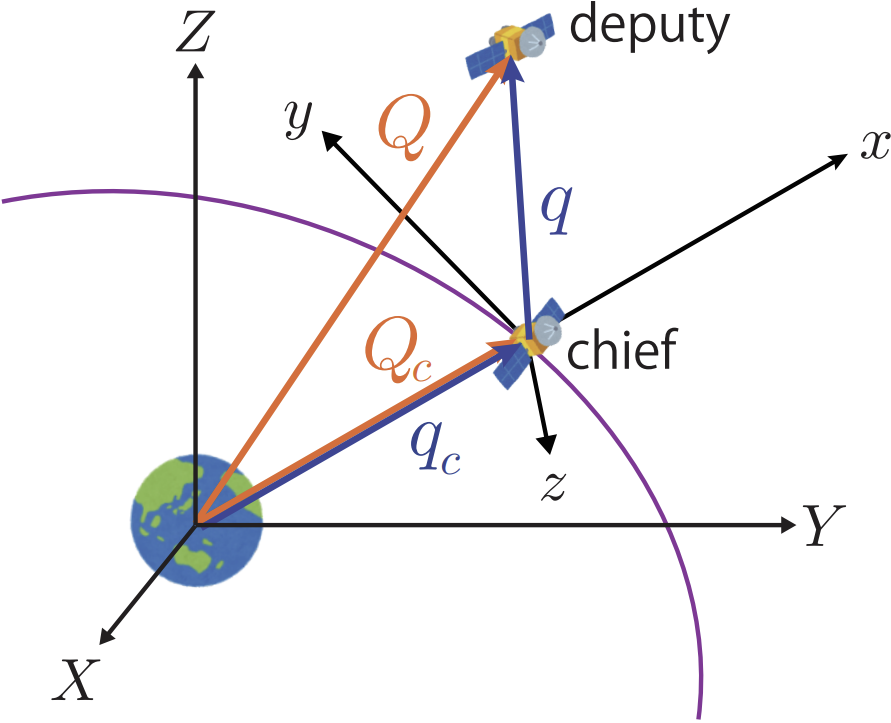}
    \caption{ECI and LVLH frames \citep{IchitoTABUCHI2024105}}
    \label{fig:lvlh}
\end{figure}
%where $x=[q^{\top}\; p^{\top}]^{\top}$. %and $y$ denotes the collocated output.

%This representation explicitly captures the interconnection, dissipation and input-output energy exchange inherent in the spacecraft relative motion.
%\subsection{Construction of the Error System via Generalized Canonical Transformations}

According to the generalized canonical transformation framework \citep{FUJIMOTO2003}, a specific transformation has been developed in \citep{IchitoTABUCHI2024105} that converts the actual state $x = [\,q^\top \; p^\top\,]^\top$ into error coordinates with respect to the desired trajectory $q^d(t)$ and its associated conjugate momentum $p^d(t)$ as follows:
\begin{equation}
\bar{x} =
\begin{bmatrix}
T(t)(q - q^d)\\[2pt]
p - p^d
\end{bmatrix}
=
\begin{bmatrix}
\bar{q}\\[2pt]
\bar{p}
\end{bmatrix},
\label{eq:errors}
\end{equation}
where \(T(t)\) is the LVLH-to-ECI transformation matrix defined previously and \( p^d = \dot{q}^d + \dot{q}_c + \omega_c^{\times}(q^d + q_c)\).
The resulting error dynamics preserve the port-Hamiltonian structure and are given by
\begin{equation}
\dot{\bar{x}} =
\begin{bmatrix}
O_{3,3} & T(t)\\[2pt]
-\,T(t)^\mathsf{T} & O_{3,3}
\end{bmatrix}
\dfrac{\partial \bar{H}(\bar{x},t)^\mathsf{T}}{\partial \bar{x}}
+
\begin{bmatrix}
O_{3,3} \\[2pt]
I_3
\end{bmatrix}\bar{u},
\label{eq:pH_error}
\end{equation}
% \begin{equation}
% \begin{cases}
% \dot{\bar{x}} =
% \begin{bmatrix}
% O_{3,3} & T(t)\\[2pt]
% -\,T(t)^\mathsf{T} & O_{3,3}
% \end{bmatrix}
% \dfrac{\partial \bar{H}(\bar{x},t)^\mathsf{T}}{\partial \bar{x}}
% +
% \begin{bmatrix}
% O_{3,3} \\[2pt]
% I_3
% \end{bmatrix}\bar{u},\\[2ex]
% \bar{y} =
% [\,O_{3,3} \;\; I_3\,]
% \dfrac{\partial \bar{H}(\bar{x},t)^\mathsf{T}}{\partial \bar{x}},
% \end{cases}
% \label{eq:pH_error}
% \end{equation}
where the transformed Hamiltonian is \(\bar{H}(\bar{x},t) = \tfrac{1}{2}\bar{p}^\top \bar{p}\).

System~\eqref{eq:pH_error} defines the port-Hamiltonian representation of the error dynamics, whose equilibrium $\bar{x}=0$ corresponds to perfect tracking of $(q,p)$ to $(q^d,p^d)$.  
System~\eqref{eq:pH_error} can be explicitly written as
\begin{equation}
\dot{\bar{q}} = T(t)\bar{p}, \qquad
\dot{\bar{p}} = \bar{u},
\label{eq:error_dynamics}
\end{equation}
which forms the basis for the sliding-mode control design presented in the next section.

\section{Control Methodology}
This section presents the control methodology for precise formation-keeping of spacecraft under bounded disturbances. We first recall a first-order Sliding Mode Control (SMC) approach used in previous work \citep{IchitoTABUCHI2024105} and discuss its limitations. Next, we introduce a boundary-layer SMC (B-SMC) strategy and analyze its stability properties, providing bounds on the sliding variable and position error. Finally, an Adaptive Boundary-layer SMC (AB-SMC) law is proposed to overcome the trade-off between transient performance and steady-state accuracy inherent in constant-gain B-SMC, and its closed-loop behavior is rigorously analyzed using Lyapunov methods.
\subsection{Preliminaries}
\begin{lem}[Stated and proved in \citep{IchitoTABUCHI2024105}]
\label{lemma1}
Consider the relative error dynamics of the spacecraft formation expressed in port-Hamiltonian form, and define the sliding variable $S$ and the control input $u$ as follows:
\begin{align}
    &S(\bar{q}, \bar{p}, t)
    = \frac{\Lambda}{2}\,T(t)^{\top}\bar{q} + \bar{p},\quad \Lambda=2\mathrm{diag}\{K_1, K_2, K_3\}\label{eq:sliding_variable}\\
    &\bar{u}
    = -k{\Lambda}\,\mathrm{sign}\!\left(S(\bar{q}, \bar{p}, t)\right)
    - \frac{\partial S(\bar{q}, \bar{p}, t)}{\partial \bar{q}}\,
      T(t)\,\bar{p}.
\label{eq:ctrl_input_sign}\end{align}
Where $K_1$, $K_2$, $K_3$, and $k$ are all positive constant gain parameters. %and $\Lambda = 2 \mathrm{diag} \{ K_1,K_2,K_3 \} \succ0$.
Under this control law, the error system reaches the sliding manifold $S = 0$ in finite time, after which it remains on the manifold and converges asymptotically to the origin.
\qed\end{lem}

%\begin{remark}
The control strategy in Lemma~\ref{lemma1} is robust to matched disturbances but can induce {chattering}, which is mitigated via a boundary-layer SMC approach in \citep{IchitoTABUCHI2024105}.
This modification confines trajectories to a {boundary layer} around the manifold rather than exactly on the sliding manifold, potentially reducing steady-state accuracy. This effect, not addressed in \citep{IchitoTABUCHI2024105}, is analysed next by deriving analytical bounds on the sliding variable and tracking error within the boundary layer.
%\end{remark}

\subsection{Boundary-Layer SMC and Error Analysis}
This subsection formalizes the analysis of the Boundary-layer SMC (B-SMC) and quantifies the impact of the boundary layer on the system performance. Specifically, it establishes finite-time convergence of the trajectories to the boundary layer and provides explicit bounds on the sliding variable and tracking error while the system evolves within this layer.

The B-SMC law is as follows:
\begin{equation}
    \bar{u}
    = -k{\Lambda}\,\mathrm{sat}_\sigma\!\left(S(\bar{q}, \bar{p}, t)\right)
    - \frac{\partial S(\bar{q}, \bar{p}, t)}{\partial \bar{q}}\,
      T(t)\,\bar{p}
\label{eq:ctrl_input_sat}\end{equation}
where $\sigma>0$ determines the thickness of the boundary layer and the saturation function $\mathrm{sat}_\sigma(\cdot)$ acts component-wise on $S$, returning $\operatorname{sign}(S_i)$ if $|S_i|>\sigma$ and $S_i/\sigma$ otherwise.

Before analysing the dynamic evolution of the error system under the B-SMC, we state a mild assumption on the closed-loop boundedness, which will be used throughout the following developments.
\begin{hypo}%[Boundedness of the closed-loop trajectories]
\label{ass:boundedness}
The closed-loop solutions of the error dynamics \eqref{eq:error_dynamics} under the considered control laws are forward complete and bounded for all \(t\geq t_0\); i.e., there exist constants $M_q,M_p>0$ such that, for all \(t\geq t_0\):
\[
    \|\bar q(t)\|_\infty\le M_q,\qquad \|\bar p(t)\|_\infty\le M_p,\qquad \|\dot T(t)^\top \,\bar q\|_\infty
\le \mu.
\]
\end{hypo}

\begin{remark}
Hypothesis~\ref{ass:boundedness} is justified by the mission scenario: during the fine formation-keeping phase, each spacecraft employs continuous low-thrust control following a preceding coarse approach based on impulsive maneuvers \citep{doi:10.2514/1.15114}. This ensures small initial errors and guarantees bounded relative dynamics. 
Moreover, at 550~km altitude the chief's angular rate is $n_c\approx 1.1\times10^{-3}\,\mathrm{rad/s}$, hence the time-variation of the ECI--LVLH transform $T(t)$ is negligible and a finite (small) $\mu$ is justified.
\end{remark}

Now, the following theorem establishes that, under the proposed B-SMC, the system trajectories reach the boundary layer in finite time and remain therein for all subsequent times.

\begin{thm}\label{thm:boundary_layer_conv}
Consider the state errors \eqref{eq:errors} with dynamics \eqref{eq:error_dynamics}.  
Apply the control law \eqref{eq:ctrl_input_sat} with sliding variable \eqref{eq:sliding_variable}, and assume Hypothesis~\ref{ass:boundedness} holds.  
Let \(\Lambda=2\operatorname{diag}\{K_1,K_2,K_3\}\) and define $K_{\min}:=\min_i K_i$, $K_{\max}:=\max_i K_i$. 
Choose \(\sigma>0\) and the control gain \(k > \frac{3}{2}\mu\frac{K_{\max}}{K_{\min}} \), where \(\mu\) is from Hypothesis~\ref{ass:boundedness}.  
Then, the closed-loop trajectories converge in finite time \(t_r<\infty\) to 
\begin{equation}
\mathcal{B}_S = \{(\overline{q},\overline{p})\in\mathbb{R}^6 : \|S(\overline{q},\overline{p},t)\|_\infty \le \sigma\},
\label{eq:boundary_layer}
\end{equation}
and remain therein for all \(t\ge t_r\).

\begin{proof}
Consider the quadratic Lyapunov function \(V = S^\top S\) with time derivatives \(\dot V = 2 S^\top \dot S\). Substituting the closed-loop dynamics \eqref{eq:error_dynamics} gives
\begin{equation}
    \begin{split}
        \dot V(S) &= 2 S^\top \Bigl(-k \Lambda \mathrm{sat}_\sigma(S) + \Lambda/2\,\dot T^\top \bar q\Bigr)\\
&\le 2 S^\top \Lambda \Bigl(-k\,\mathrm{sat}_\sigma(S) + \frac{\mu}{2}\mathbf{1}\Bigr),
    \end{split}
\end{equation}
where $\mathbf{1}=(1,1,1)$ and Hypothesis~\ref{ass:boundedness} is used. 
If at least one component \(|S_i|>\sigma\), then \(S^\top \mathrm{sat}_\sigma(S) \ge \|S\|_\infty\) and
\[
S^\top \Lambda \mathrm{sat}_\sigma(S) \ge K_{\min} \|S\|_\infty, \quad
S^\top \Lambda \frac{\mu}{2} \mathbf{1} \le \frac{3\mu}{2} K_{\max} \|S\|_\infty.
\]
Combining yields
\begin{equation}
    \dot V(S) \le 2 \Big(-k K_{\min} + \tfrac{3\mu}{2} K_{\max}\Big) \|S\|_\infty =: -\alpha \|S\|_\infty,
    \label{eq:dot_Vs_last}
\end{equation}
with \(\alpha>0\) by the above choice of $k$. By equivalence of norms in \(\mathbb{R}^3\), \(\dot V(S) \le -\tilde\alpha \|S\| = -\tilde\alpha V^{1/2}\) with $\tilde{\alpha}=\alpha/\sqrt{3}$, guaranteeing finite-time convergence to
\(\mathcal{B}_S = \{(q,p): \|S\|_\infty \le \sigma\}\).
\end{proof}
\end{thm}

From Theorem~\ref{thm:boundary_layer_conv}, there exists a finite time $t_r<\infty$ such that \(\|S(t)\|_\infty \le \sigma\) \(\forall t \ge t_r.\)
Thus, for $t \ge t_r$ (sliding phase), the trajectories remain within the boundary layer, 
and the system evolves according to the dynamics induced by the sliding variable \eqref{eq:sliding_variable}. 
The convergence properties of the system in this phase are addressed in the following theorem.

\begin{thm}\label{thm:sliding_phase_conv}
Consider the state errors \eqref{eq:errors} with dynamics \eqref{eq:error_dynamics}, and assume that the trajectories evolve within the boundary-layer set \(\mathcal{B}_S\) defined in Theorem~\ref{thm:boundary_layer_conv} for any $t\geq t_r$. 
If the gains of the sliding variable \eqref{eq:sliding_variable} satisfy \(K_1,K_2,K_3>0\), then the state error \(\bar q\) converges asymptotically to the
set
\[
\mathcal{B}_q=\{\bar q \in \mathbb{R}^3 : \|\bar q(t)\|_\infty \le 2\sqrt{3}\sigma / K_{\min} \}, \quad 
K_{\min} := \min_i K_i.
\]

\begin{proof}
Consider the Lyapunov function \(V(\bar q) = \bar q^\top \bar q\), with time derivative
\(\dot V = 2 \bar q^\top \dot{\bar q}\) and \(\dot{\bar q} = T(t)\bar{p}\).  
From \eqref{eq:sliding_variable}, we have \(\bar p = S - \frac{\Lambda}{2} T^\top \bar q\), so that the \(\bar q\)-dynamics in the sliding phase reads
\[
\dot{\bar q} = T(t) S - \frac{1}{2} T(t) \Lambda T(t)^\top \bar q, \quad \|S\|_\infty \le \sigma.
\]
Exploiting the orthogonality of \(T\) (\(T^\top T = I\)) and the relation \(\|S\| \le \sqrt{3} \|S\|_\infty\), the Lyapunov derivative satisfies
\[
\dot V \le 2 \sqrt{3}\, \sigma \|\bar q\| - K_{\min} \|\bar q\|^2.
\]
Therefore, whenever \(\|\bar q\|
> 2\sqrt{3} \sigma / K_{\min}\), the quadratic term dominates and \(\dot V < 0\).  
By Lyapunov’s direct method, \(\bar q(t)\) converges asymptotically to the set \(\mathcal{B}_q \)
due to $\|\overline q\|_\infty\leq\|\overline q\|$, which concludes the proof.
\end{proof}
\end{thm}

Theorem~\ref{thm:sliding_phase_conv} shows that, once trajectories enter the boundary layer $\mathcal{B}_S$, the position error $\bar q$ is bounded by $\sigma/K_{\min}$. Reducing $\sigma$ is limited by chattering, so improving steady-state accuracy requires increasing $K$. However, high constant gains increase control aggressiveness for large position errors, creating a trade-off between disturbance rejection and actuator saturation.

\subsection{Boundary-Layer SMC with Adaptive Sliding Manifold}\label{sec:AB-SMC}
To overcome the limitations of the constant-gain B-SMC discussed above, this paper proposes an Adaptive Boundary-layer SMC (AB-SMC) strategy. First, the sliding variable is defined with a scalar gain $K(t)$, identical for all components.
\begin{equation}
S(\bar q, \bar p, t) = K(t)\,T(t)^\top \bar q + \bar p, 
\quad \Lambda = 2 K(t) I_3,
\label{eq:sliding_variable_scalar}\end{equation}
Then, $K(t)$ is adjusted in real-time according to the following adaptive law: 
\begin{equation}
    \dot K = \mathrm{Proj}_{\mathcal K} \Big( \eta(\overline K - K) - \gamma h_Q(\bar q) \|\bar q\|_\infty^2 \Big),
\label{eq:dot_K}\end{equation}
where \(\eta,\gamma>0\) are design parameters and $\mathcal{K}=[\underline{K},\overline{K}]$ is the adaptation interval within which $K_i$ is forced to evolve. Indeed, the $Proj$ function works as a saturated integrator \cite{pomet1992adaptive}:
\begin{equation}\begin{split}
&\underset{{K} \in \mathcal{K}}{\text{Proj}}=\begin{cases}
        \text{max}\{0,f\}\quad&\text{if}~K=\underline{K}\\
        f\quad&\text{if}~K\in\text{int}(\mathcal{K})\\
        \text{min}\{0,f\}\quad&\text{if}~K=\overline{K}
    \end{cases},\\
&f=\eta\left(\overline{K}-K \right)-\gamma h_Q(\bar q)\|\overline{q}\|_\infty^2.
\label{eq:proj_nth}\end{split}\end{equation}
Then, $h_Q$ is as follows:
\begin{equation}\begin{split}
&h_Q(|\overline{q}|)=\begin{cases}
1,& \text{if } \|\bar q\|_\infty>Q,\\
0,& \text{if } \|\bar q\|_\infty\le Q,
\end{cases}
\label{eq:h_q}\end{split}\end{equation}
and $Q>0$ is a selectable parameter. The adaptive law \eqref{eq:dot_K}-\eqref{eq:h_q} governs the evolution of the sliding gain $K(t)$ according to the instantaneous tracking error $\bar q$: gains decrease for large errors to limit control effort and increase for small errors to improve steady-state performance. This behaviour is rigorously characterized in the following proposition.

\begin{proposition}\label{prop:adaptive_gain_behavior}
Consider the adaptive law \eqref{eq:dot_K}–\eqref{eq:h_q}, with parameters $Q>0$, $\eta>0$, and 
\begin{equation}
    \gamma = \eta(\overline{K}-\underline{K})/Q^2. 
\label{eq:gamma}\end{equation}
Then, the evolution of \(K(t)\) satisfies:
\begin{enumerate}
    \item If $\|\bar q\|_\infty>Q$, then $K(t)$ converges in finite time to the lower bound $\underline{K}$.
    \item If $\|\bar q\|_\infty\le Q$, then $K(t)$ converges exponentially to the upper bound $\overline{K}$.
\end{enumerate}

\begin{proof}
Consider first the case $\|\bar q\|_\infty>Q$. Define the Lyapunov candidate
\(V = (K - \underline{K})^2.\)
For $K \in (\underline K, \overline K)$, the projection operator is inactive, and the time derivative of the Lyapunov candidate along the trajectories of \eqref{eq:dot_K}–\eqref{eq:h_q} is
\begin{equation}\begin{split}
    \dot V &= 2(K - \underline K) \dot K = 2(K - \underline K) \, \Big(\eta(\overline K - K) - \gamma \|\bar q\|_\infty^2 \Big)\\
    & \leq -2\eta(\overline K - \underline K)\left(\frac{\|\bar q\|_\infty^2}{Q^2}-1\right)(K - \underline K):=-\alpha(K - \underline K)
\end{split}\end{equation}
with $\alpha>0$ due to $\|\bar q\|_\infty>Q$. Therefore, \(\dot V < -\alpha V^\frac{1}{2}\) whenever $K>\underline K$. 
By Lyapunov’s direct method, $K$ reaches $\underline K$ in finite time.

Next, for $\|\bar q\|_\infty\le Q$, the adaptive law reduces to
\[
\dot K = \eta (\overline K - K), \quad K \in (\underline K, \overline K),
\]
whose solution is
\[
K(t) = \overline K - (\overline K - K(t_0)) e^{-\eta (t-t_0)}.
\]
This shows that $K(t)$ converges exponentially to $\overline K$ with a convergence rate given by \(\eta\).
\end{proof}
\end{proposition}
The following remarks highlight two key points arising from the above proposition, which will be instrumental in the subsequent stability analysis.

\begin{remark}%[Boundedness of the gain derivative]
Under Hypothesis~\ref{ass:boundedness}, $\bar q(t)$ is bounded by $M_q$, hence from \eqref{eq:dot_K}--\eqref{eq:h_q} one has
\[
|\dot K(t)| \le \eta(\overline K-\underline K)+\gamma M_q^2 := D_{\max}, \quad \forall t\ge t_0.
\]
Thus the adaptive gain evolves with uniformly bounded rate.
\label{remark_kmax}\end{remark}

\begin{remark}%[Finite duration of the adaptive transient]
From Proposition~\ref{prop:adaptive_gain_behavior}, $K(t)$ reaches one boundary of $\mathcal K=[\underline K,\overline K]$ in finite time; thereafter, the projection operator saturates the update and sets $\dot K(t)=0$. Hence the adaptation is active only over a finite interval, and $\dot K(t)\!\to\!0$ in finite time.
\label{remark_dotK_2_0}\end{remark}

We are now ready to state the main stability theorem for the proposed AB-SMC, characterizing the convergence properties of the closed-loop system.

\begin{thm}%[Convergence under adaptive B-SMC]
\label{thm:two_phase}
Consider the state errors \eqref{eq:errors} with dynamics \eqref{eq:error_dynamics},    controlled by the AB-SMC with control law \eqref{eq:ctrl_input_sat}, sliding variable \eqref{eq:sliding_variable_scalar}, and the adaptive law \eqref{eq:dot_K}--\eqref{eq:h_q}. 
Assume Hypothesis~\ref{ass:boundedness} and select $k>\frac{3}{2}\mu$ and $\gamma$ as in Eq. \eqref{eq:gamma}. Let $\sigma>0$ denote the boundary-layer thickness, and $\mathcal K=[\underline K,\overline K]$ the admissible interval for the sliding gain $K(t)$, with \(\overline{K}>\underline{K}>0\).   
Then the closed-loop exhibits the following two phases.

\begin{enumerate}
\item[\((\mathrm{P}1)\)] (\emph{Coarse phase:} $\|\bar q\|_\infty>Q$)  
If initially the tracking error satisfies $\|\bar q(t_0)\|_\infty>Q$, then $K(t)$ converges in finite time to $\underline K$ and the trajectories enter the boundary layer surrounding the sliding manifold parametrized by $\underline K$. While sliding on this manifold, the position error converges asymptotically to
\begin{equation}
\mathcal B_q^{\underline K}
= \Big\{\bar q\in\mathbb R^3:\;\|\bar q\|_\infty \le \frac{2\sqrt{3}\sigma}{\underline K}\Big\}.
    \label{eq:coarse_set}
\end{equation}
\item[\((\mathrm{P}2)\)] (\emph{Fine phase:} $\|\bar q\|_\infty\le Q$)  
If the tracking error satisfies $\|\bar q\|_\infty\le Q$ (either initially or after (P1)), then the gains evolve to $\overline K$ (exponentially) and in finite time the trajectories enter the boundary layer surrounding the sliding manifold parametrized by $\overline K$. While sliding on this manifold, the position error converges asymptotically to
\begin{equation}
\mathcal B_q^{\overline K}
= \Big\{\bar q\in\mathbb R^3:\;\|\bar q\|_\infty \le \frac{2\sqrt{3}\sigma}{\overline K}\Big\}.
 \label{eq:fine_set}
\end{equation}
\end{enumerate}

Moreover, if the threshold is chosen as
\begin{equation}
Q=\frac{2\sqrt{3}\sigma}{\underline K},
\label{eq:transition}
\end{equation}
then the attainment of (P1) (i.e. convergence to $\mathcal B_q^{\underline K}$) implies that the condition for (P2) is satisfied, hence the closed-loop necessarily transitions from (P1) to (P2).

\begin{proof}
1. (\emph{Coarse phase, $\|\bar q(t_0)\|_\infty>Q$})  
Consider the Lyapunov function \(V=S^\top S\). Its derivative under \eqref{eq:ctrl_input_sat}, \eqref{eq:sliding_variable_scalar}, \eqref{eq:dot_K} reads
\begin{equation}
\dot V = 2 S^\top\Big(-k\Lambda\,\mathrm{sat}_\sigma(S) + \tfrac{\Lambda}{2}\dot T^\top \bar q\Big)
  + 2 S^\top(\dot K\,T^\top\bar q),
\label{eq:dotS_general}
\end{equation}
where the last term is of indefinite sign but uniformly bounded (Hypothesis~\ref{ass:boundedness}, Remark~\ref{remark_kmax}), ensuring \(\|S\|\) remains finite during the $K$-adaptation.  
Proposition~\ref{prop:adaptive_gain_behavior} and Remark~\ref{remark_dotK_2_0} guarantee that \(K\to \underline K\) in finite time \(t_K\), after which \(\dot K=0\) and \(\dot V\) reduces to the constant-gain form \eqref{eq:dot_Vs_last}.  
Theorems~\ref{thm:boundary_layer_conv}–\ref{thm:sliding_phase_conv} then ensure finite-time entry into the boundary layer and asymptotic convergence to \(\mathcal B_q^{\underline K}\).

2. (\emph{Fine phase, $\|\bar q\|_\infty\le Q$})\\
The arguments are analogous to those in the coarse phase. Here \(K(t)\to \overline K\) exponentially (Proposition~\ref{prop:adaptive_gain_behavior}).  
Considering the Lyapunov function \(V=S^\top S\), the indefinite term in \(\dot V\) (Eq.~\ref{eq:dotS_general}) is bounded and decays exponentially along with \(\dot K\), becoming practically negligible after a finite transient.  
Thereafter, \(\dot V\) reduces to the constant-$K$ form (Eq.~\ref{eq:dot_Vs_last}), so Theorems~\ref{thm:boundary_layer_conv}–\ref{thm:sliding_phase_conv} guarantee finite-time entry into the boundary layer and asymptotic convergence to \(\mathcal B_q^{\overline K}\).

3. (\emph{Phase transition})  
Choosing $Q = 2\sqrt{3}\sigma/\underline K$ aligns the bound of $\mathcal B_q^{\underline K}$ with the activation threshold, so convergence to $\mathcal B_q^{\underline K}$ (P1) implies $\|\bar q\|_\infty\le Q$, which triggers the adaptive increase $K\to\overline K$ and the transition to (P2).
\qed
\end{proof}

\end{thm}
 
\textcolor{black}{\begin{remark}
The switching function $h_Q(\bar q)$ defined in Eq. \eqref{eq:h_q} can be endowed with a hysteresis mechanism to improve robustness with respect to measurement noise and non-monotonic disturbances. Specifically, instead of a single threshold $Q$, two thresholds $Q_{\mathrm{off}}<Q_{\mathrm{on}}$ are introduced and $h_Q$ can be defined as a binary variable with memory according to
\[
h_Q(t^+)=
\begin{cases}
0, & \|\bar q(t)\|_\infty \le Q_{\mathrm{off}},\\
1, & \|\bar q(t)\|_\infty \ge Q_{\mathrm{on}},\\
h_Q(t), & Q_{\mathrm{off}} < \|\bar q(t)\|_\infty < Q_{\mathrm{on}}.
\end{cases}
\]
Accordingly, the term $\gamma h_Q(\bar q)\|\bar q\|_\infty^2$ in the adaptive law \eqref{eq:dot_K} is active for $\|\bar q\|_\infty \ge Q_{\mathrm{on}}$ and inactive for $\|\bar q\|_\infty \le Q_{\mathrm{off}}$, while retaining its previous value otherwise.
In particular, choosing $Q_{\mathrm{off}}=Q$ preserves the transition condition between Phase~(P1) and Phase~(P2) stated in Theorem~6. Furthermore, for any time interval where $h_Q(t)$ is constant, the closed-loop system reduces to either Phase (P1) or Phase (P2), and the conclusions of Theorem~6 apply on such intervals. Hysteresis mechanisms of this kind are standard in adaptive and switching control to avoid chattering phenomena due to small oscillations around the nominal threshold; see, e.g., \citep{159571, 2005ICSys}.
\end{remark}
\begin{remark}
    The proposed adaptive law (12)–(15) requires the selection of three design parameters: $\eta$,  $\gamma$, and $Q$. While Eqs. (15) and (19) provide explicit relationships to select $\gamma$ and $Q$ once $\eta$ is chosen, the choice of the adaptation rate $\eta$ itself deserves special attention.
    As shown in Proposition 5, the parameter $\eta$ sets the convergence rate of $K(t)$ toward $\overline{K}$ during the fine phase (P2). On the other hand, Theorem 6 shows that, during the adaptation transient, the nonzero term $\dot{K}(t)$ introduces additional perturbative terms in the $S-$dynamics. Since the magnitude of these terms increases with $\eta$, this should be selected as a compromise between adaptation speed and robustness of the sliding motion: smaller values of $\eta$ reduce the perturbations induced by $\dot{K}(t)$ and favor the preservation of sliding during the transient, whereas larger values of $\eta$ yield faster convergence of $K(t)$ at the expense of a more pronounced transient disturbance.
%Specifically, $\eta$ governs the speed at which the adaptive gain $K(t)$ converges to $\overline K$ during the \textit{Fine phase} (P2) as defined in Theorem 6. As discussed in point 2. of the Proof of Theorem 6, the time variation of $K(t)$ introduces additional perturbative terms in the Lyapunov function derivative during this transient. Smaller values of $\eta$ reduce the magnitude of these terms, promoting the persistence of sliding motion on the time-varying sliding manifold, while larger values of $\eta$ accelerate the convergence of $K(t)$ but may temporarily violate the reaching condition. Therefore, $\eta$ should be selected to balance robustness during the adaptation transient with the desired speed of convergence of the adaptive gain $K(t)$.
\end{remark}}

The proposed AB-SMC strategy overcomes the trade-off inherent in constant-gain B-SMC between fast transient response and steady-state precision. By adjusting the sliding gain $K(t)$ according to the tracking error, the AB-SMC achieves both a smoother transient for large errors and improved steady-state accuracy for small errors. The effectiveness of this approach will be demonstrated in the numerical simulation results in Section~\ref{sec:simulations}.

\section{Simulation and Results}\label{sec:simulations}

This section presents the numerical validation of the proposed AB--SMC strategy for satellite formation keeping. The performance of the AB--SMC is evaluated through detailed time-domain simulations and compared with three alternative control laws. 
The simulations are designed to assess the controllers in terms of transient response, steady-state accuracy, chattering behaviour, and robustness under perturbed initial conditions. 
The following subsections describe the simulation setup, the time-domain results, an analysis of chattering effects, tuning strategies for the B--SMC, and a Monte Carlo campaign to statistically evaluate controller performance.

\textcolor{black}{\subsection{Orbital Simulator Framework}
The orbital simulation framework models the relative motion of multiple spacecraft in LEO while explicitly accounting for modeling errors and uncertainties. The orbital simulator is derived from the analysis described in  \citep{SILVIA25}. 
In the simulator, nonlinear orbital dynamics are propagated including dominant perturbations such as $J_2$ effect, atmospheric drag, and solar radiation pressure.
%This simulator enables performance assessment and validation of formation-flying control strategies under realistic operational conditions relevant to space-based interferometry missions.
Actuator uncertainties are accounted by considering saturation constraints, minimum thrust limitations, and noise effects. In particular, the thruster noise is modeled as zero-mean Gaussian white noise with an amplitude spectral density (ASD)
of $1~\mu\mathrm{N}/\sqrt{\mathrm{Hz}}$ over a bandwidth of $B = 1~\mathrm{Hz}$.
Under this assumption, the root-mean-square (RMS) thrust noise over the specified bandwidth is
\begin{equation}
F_{\mathrm{RMS}} = F_{\mathrm{ASD}} \sqrt{B}
= 1 \times 10^{-6}~\mathrm{N},
\end{equation}
which corresponds to an acceleration noise of $a_{\mathrm{RMS}}=5.5\times10^{-8}~\text{m/s}^2$ for a spacecraft mass of $m = 180~\mathrm{kg}$. 
Then, the continuous-time white noise model is discretized by scaling the RMS acceleration noise as
\begin{equation}
a_{\mathrm{noise}} = \frac{a_{\mathrm{RMS}}}{\sqrt{\Delta t}} \, \mathcal{N}(0,1),
\end{equation}
where $\mathcal{N}(0,1)$ is a standard normal random variable and $\Delta t$ denote the discrete simulation time step. In simulations, the acceleration noise term $a_{\mathrm{noise}}$ is added to the nominal acceleration command, while an upper saturation limit is applied to model the maximum available thrust.
%For a spacecraft mass of $m = 180~\mathrm{kg}$, the corresponding RMS acceleration noise is
% \begin{equation}
% a_{\mathrm{RMS}} = \frac{F_{\mathrm{RMS}}}{m}
% = \frac{1 \times 10^{-6}}{180}~\mathrm{m/s^2}.
% \end{equation}
In addition, the thrusters are assumed to exhibit a minimum achievable thrust $f_{\min}=10~\mu\mathrm{N}$, which corresponds to a minimum acceleration of $a_{min} = 5.5\times10^{-7}~\text{m/s}^2$.
Thrust commands below this level cannot be reliably produced and are therefore neglected. In simulations, when the magnitude of the commanded acceleration falls below $a_{\min}$, the command is clipped to zero.}
% In this study, the minimum thrust is set to
% \begin{equation}
% f_{\min} = 10~\mu\mathrm{N},
% \end{equation}
% which corresponds to a minimum achievable acceleration of
% \begin{equation}
% a_{\min} = \frac{f_{\min}}{m}
% = \frac{10 \times 10^{-6}}{180}~\mathrm{m/s^2}.
% \end{equation}
%\subsection{Simulation Implementation Details}
% Let $\Delta t$ denote the discrete simulation time step.
% The continuous-time white noise model is discretized by scaling the RMS acceleration noise as
% \begin{equation}
% a_{\mathrm{noise}} = \frac{a_{\mathrm{RMS}}}{\sqrt{\Delta t}} \, \mathcal{N}(0,1),
% \end{equation}
% where $\mathcal{N}(0,1)$ is a standard normal random variable.
% %In MATLAB, this term is generated using the \texttt{randn} function.
% The acceleration noise term $a_{\mathrm{noise}}$ is added to the nominal acceleration command.
% Furthermore, commanded accelerations are subject to actuator constraints:
% an upper saturation limit is applied to model the maximum available thrust.
% while a lower deadband is enforced to account for the minimum thrust capability.
% Specifically, when the magnitude of the commanded acceleration falls below $a_{\min}$, the command is clipped to zero.}

\subsection{Simulation Setup}
The effectiveness of the proposed AB-SMC (Eq.~\ref{eq:ctrl_input_sat}, \ref{eq:sliding_variable_scalar}-\ref{eq:h_q}) is verified by numerical simulations and compared with the passivity-based controller (PBC) from ~\citep{SH26, IchitoTABUCHI2024105}, the classical SMC (Eq.~\ref{eq:ctrl_input_sign}), and the constant-gain B-SMC (Eq.~\ref{eq:ctrl_input_sat}). 
The parameters of the controllers used in the simulations are summarized in Table~\ref{tab:control_params}, where it is implied that the three sliding mode controllers employ the sliding surface \eqref{eq:sliding_variable_scalar} with a unique scalar gain for all components.
\begin{table}[htbp]\label{tab:control_params}
\centering
\caption{Controller parameters}
\begin{tabular}{ll}
\toprule
\textbf{Controller} & \textbf{Parameters} \\
\midrule
PBC & $K_p = 0.1$, $K_d = 1.0$ \\
SMC & $k = \tfrac{5}{\sqrt{3}}\times10^{-3}$, $K = 0.1$ \\
B--SMC & $k = \tfrac{5}{\sqrt{3}}\times10^{-3}$, $\sigma = k/5$, $K = 0.2$ \\
AB--SMC & $k = \tfrac{5}{\sqrt{3}}\times10^{-3}$, $\sigma = k/5$, $\eta = 0.05$, $\overline{K} = 0.2$, \\ 
& $\underline{K}$ s.t. $Q = 4\times10^{-3}$~m~\text{(Eq. \ref{eq:transition})}\\
\bottomrule
\end{tabular}
\end{table}

The target formation maintains an equilateral triangle with 100~m side length in a circular orbit at 550~km altitude, corresponding to a relative distance $\rho = 57.735$~m. The chief’s orbital elements are \(a_c = 6928.137~\text{km}\), \( e_c = 0\), \(i_c = 98^{\circ}\), \(\Omega_c = 0^{\circ}\), 
while the desired relative trajectory for the %$i-$th deputy ($i=1,2,3$) is 
deputies are \(q_{d_1} =\frac{1}{2}\rho \sin(n_c t - \frac{\pi}{2})\), \(q_{d_2} = \rho \cos(n_c t + \frac{\pi}{6})\), and \(q_{d_3}=\frac{\sqrt{3}}{2}\rho \sin(n_c t + \frac{5\pi}{6})\).
% \begin{equation}
% \label{eq:qd}
% q_{d_i} =
% \begin{bmatrix}
% \frac{1}{2}\rho \sin(n_c t + \phi_i) \\
% \rho \cos(n_c t + \phi_i) \\
% \frac{\sqrt{3}}{2}\rho \sin(n_c t + \phi_i)
% \end{bmatrix},
% \quad
% \boldsymbol{\phi}=\left[-\frac{\pi}{2},\,\frac{\pi}{6},\,\frac{5\pi}{6}\right].
% \end{equation}
An initial position error of 0.25~m (0.5\%) is applied, and $\rho$ in the formulation of $q_{d}$ %from Eq.~\eqref{eq:qd} 
is replaced by $1.005\rho$ for deputies~1 and~2 and by $0.995\rho$ for deputy~3. A micro-thruster is used as actuator, with $\|u\|_{\max}=5\times10^{-6}$~m/s$^2$.
Simulations are performed in MATLAB using a 4th-order Runge–Kutta integrator with 0.1~s. The simulator propagates the formation dynamics over time by integrating Eq.~\eqref{eq:ECI_dynamics} 
from the initial conditions described above, while the control input can be selected among the four controllers listed in Table~\ref{tab:control_params}.
\begin{figure}
    \centering
    \includegraphics[width=0.75\linewidth]{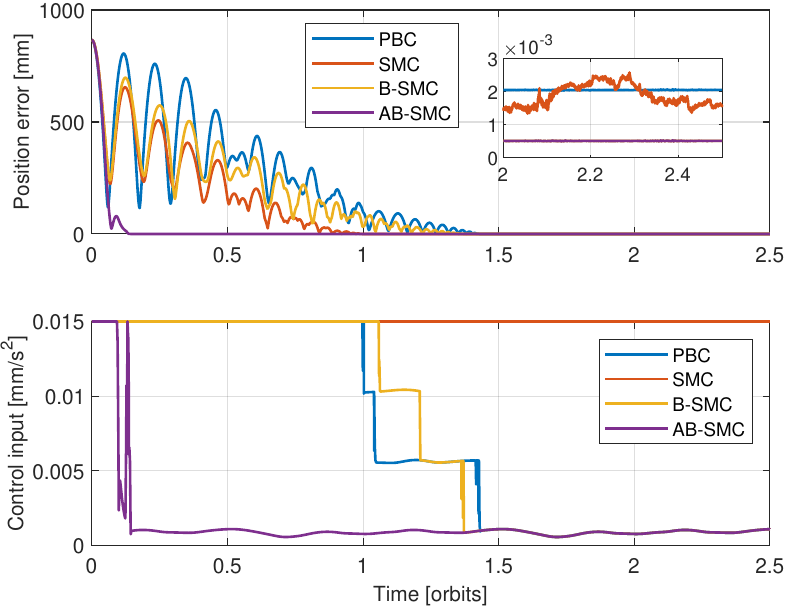}
    \caption{Global formation error (top) and total control input (bottom) 
for the control laws in Table \ref{tab:control_params}.}
    \label{fig:timeseries}
\end{figure}

\subsection{Time-Domain Performance}\label{sec:time-domain-res}
Figure~\ref{fig:timeseries} illustrates the time evolution of the total formation error 
(upper plot) and the total control effort (lower plot) for the four controllers. 
The global formation error is defined as 
$e_{\mathrm{tot}} = \sum_{i=1}^{3} \|q_i - q_{d_i}\|$, 
while the total control input is 
$u_{\mathrm{tot}} = \sum_{i=1}^{3} \|u_i\|$ (i=1,2,3 {deputies}). 
It can be observed that the proposed AB--SMC achieves a significantly faster convergence 
to the desired formation compared to the other controllers. 
Moreover, the steady-state error, reported in the inset box for half an orbit, 
is smaller then both the SMC and the PBC. 
As expected, it matches exactly the steady-state performance of the B--SMC, 
since both share the same constant upper gain $\overline{K}$ in steady state.
Furthermore, the bottom panel of Fig.~\ref{fig:timeseries} indicates that the AB--SMC achieves this improved performance while demanding the lowest overall control effort among all compared controllers.
\begin{figure}
    \centering
    \includegraphics[width=0.8\linewidth]{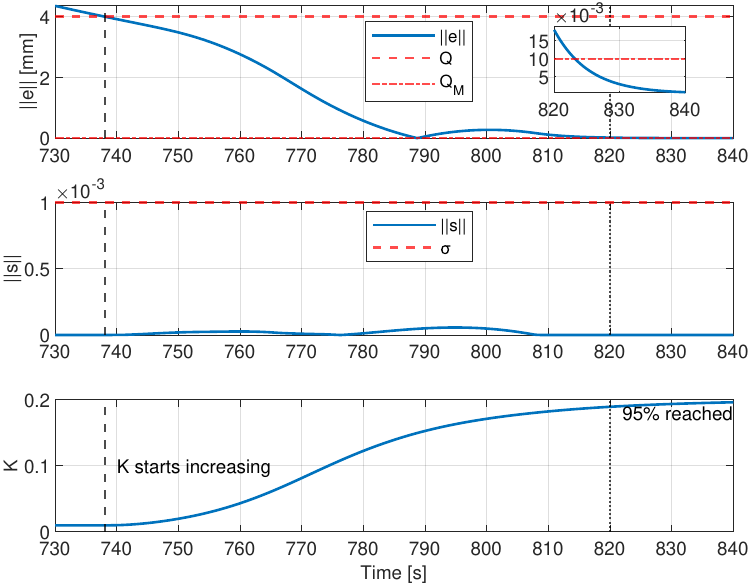}
    \caption{Representative evolution of the third deputy under the AB--SMC: 
tracking error $\|e\|_\infty$, sliding variable $\|s\|_\infty$, and adaptive gain $K(t)$.}
    \label{fig:proof}
\end{figure}
Figure~\ref{fig:proof} shows a representative time segment of the third deputy’s dynamics under the proposed AB--SMC, illustrating and validating the theoretical results presented in Section~\ref{sec:AB-SMC}.
When the tracking error satisfies $\|e\|_\infty \le Q$, i.e. $\bar{q}\in\mathcal B_q^{\underline K}$ in Eq. \eqref{eq:coarse_set}, 
the adaptive law activates and the gain $K(t)$ starts increasing. 
During this adaptation transient, the sliding variable \(S\) exhibits small oscillations, 
although it remains inside the boundary layer (Eq.~\ref{eq:boundary_layer}). 
As the adaptation subsides, the gain \(K(t)\) reaches its steady-state value \(\overline{K}\), 
and the oscillations of the sliding variable gradually vanish.
After this adaptation phase, the tracking error is confined within the predicted set  $Q_M = 2\sqrt{3}\,\sigma/\overline{K}$, confirming the asymptotic bound in Eq.~\eqref{eq:fine_set}.

\subsection{Chattering Analysis}
To quantify the chattering behavior of the controllers, we compute the total variation (TV) of a generic signal \(\xi\) as follows:
\[
\text{TV} = \sum_k |\xi_{k+1} - \xi_k|,
\]
where $k$ indexes the sampling instants. In our analysis, \(\xi\) represents either the position tracking error or the control input, computed separately for the $x$, $y$, and $z$ axes of each deputy and then summed over the three deputies.
Table~\ref{tab:chattering} summarizes the results computed during steady-state, which is defined as the period when the position error remains below $10^{-3}$~m.
\begin{table}[htbp]
\centering
\caption{Total variation of position error and control input (chattering measure)}
\label{tab:chattering}
\begin{tabular}{lcc}
\toprule
\textbf{Controller} & \textbf{TV Position Error} & \textbf{TV Control Input} \\
 & $[mm]$ & $[mm/s^2]$ \\
\midrule
PBC & $9.26\times10^{-2 }$  & $3.12\times10^{-2}$ \\
SMC & $2.15$ & $6.27\times10^{2}$ \\
B--SMC & $9.56\times10^{-2}$ & $5.14\times10^{-2}$ \\
AB--SMC & $1.99\times10^{-1}$ & $1.07\times10^{-1}$ \\
\bottomrule
\end{tabular}
\end{table}
It can be observed that the classical SMC exhibits significant chattering, especially on the control input, as expected. The B--SMC effectively suppresses this effect, reducing the total variation to the same order of magnitude as the PBC. The proposed AB--SMC provides a similar level of chattering reduction, slightly higher than the B--SMC, but still negligible, confirming that the adaptive mechanism does not introduce any practical chattering issues.

\subsection{B--SMC Tuning Analysis}
As observed in Section \ref{sec:time-domain-res}, the constant-gain B--SMC with a high gain $\overline{K}$ leads to unsatisfactory transient performance, despite providing good steady-state accuracy. To address this issue and achieve both satisfactory transient behavior and steady-state performance without requiring adaptation, we propose an alternative tuning strategy for the B--SMC. In this approach, the control gain is set to its lower bound, $K = \underline{K}$, which, as previously noted, ensures favorable dynamic behavior, while the boundary layer thickness $\sigma$ is reduced to improve steady-state accuracy.
%The influence of the boundary layer thickness $\sigma$ and the constant gain $K$ on the performance of the B--SMC is investigated. In particular, we consider a tuning strategy with $K = \underline{K}$ and progressively smaller values of $\sigma$ to improve steady-state accuracy. 
Table~\ref{tab:bs_mc_tuning} summarizes the results computed during steady-state, defined as the period when the position error remains below $10^{-1}$~m. This threshold ensures a meaningful comparison across all controllers. The table also includes results for the standard B--SMC with $K = \overline{K}$ and for the AB--SMC.
Table~\ref{tab:bs_mc_tuning} shows that reducing the boundary layer thickness $\sigma$ generally improves steady-state precision. However, for very small $\sigma$ (last row), the control input exhibits a sharp increase in total variation, indicating pronounced chattering. This chattering not only increases actuator activity but also compromises closed-loop precision, preventing the controller from reaching the accuracy levels achieved by the AB--SMC or the B--SMC with $K = \overline{K}$. These results confirm that the adaptive strategy can provide performance that is unattainable with a constant-gain B--SMC.

\begin{table}[htbp]
\centering
\caption{RMSE of position error and total variation of control input in steady-state}
\label{tab:bs_mc_tuning}
\resizebox{\columnwidth}{!}{
\begin{tabular}{lcc}
\toprule
\textbf{Controller / Tuning} & \textbf{RMSE} \((mm)\) & \textbf{TV Input} \((mm/s^2)\) \\
\midrule
B--SMC ($K=\overline{K}$, $\sigma$) & $1.55\cdot10^{-3}$ & $5.87\cdot10^{-2}$ \\
AB--SMC & $1.51\cdot10^{-3}$ & $1.11\cdot10^{-1}$ \\
B--SMC ($K=\underline{K}$, $\sigma*0.05$) & $1.18\cdot10^{-2}$ & $6.52\cdot10^{-2}$ \\
B--SMC ($K=\underline{K}$, $\sigma*0.01$) & $6.43\cdot10^{-3}$ & $6.89\cdot10^{-2}$ \\
B--SMC ($K=\underline{K}$, $\sigma*0.005$) & $6.21\cdot10^{-3}$ & $8.74\cdot10^{-2}$ \\
B--SMC ($K=\underline{K}$, $\sigma*0.001$) & $2.00\cdot10^{-2}$ & $1.36$ \\
\bottomrule
\end{tabular}}
\end{table}

\subsection{Monte Carlo analysis}
A Monte Carlo campaign was performed to statistically evaluate the robustness and performance of the proposed control laws under perturbed initial conditions. 
For each run, the initial position of every deputy were randomly sampled within a hypersphere of radius 0.5~m centered on the nominal GCO state, using uniform random sampling. 
All simulations were propagated for 2.5 orbital periods, and all controllers successfully achieved formation convergence (100\% success rate). 
The results are summarized in Figure~\ref{fig:MC_results}, which compares the four controllers in terms of settling time, total control effort, and RMSE of the global formation position error. 
The settling time is defined as the instant when the global formation error falls below $10^{-3}$~m, while the control effort is computed by time integration of the total control input using the \texttt{trapz} function. 
Each boxplot represents the statistical distribution across all Monte Carlo runs: the central line indicates the median, the box edges denote the first and third quartiles, and the whiskers extend to the most extreme data points not considered outliers. 
As shown in Figure~\ref{fig:MC_results}, the proposed AB--SMC consistently outperforms the other controllers, exhibiting the smallest median settling time, the lowest control effort, and the minimum RMSE of the formation position error.
\begin{figure}
    \centering
    \includegraphics[width=0.7\linewidth]{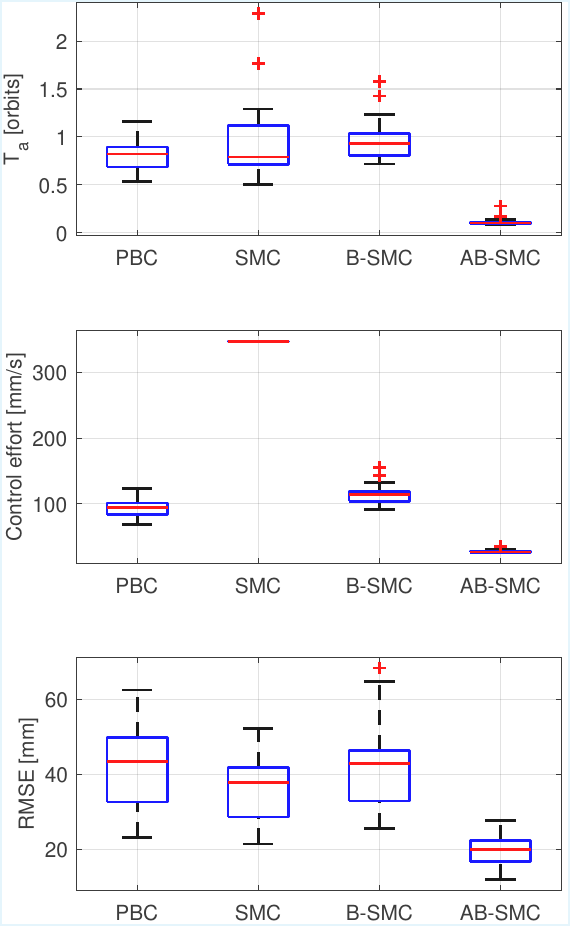}
    \caption{Statistical comparison of the controllers in Table \ref{tab:control_params} in terms of settling time, total control effort, and RMSE of the global formation position error across all Monte Carlo runs.}
    \label{fig:MC_results}
\end{figure}
Then, Table~\ref{tab:MC_RMSE} reports the median and standard deviation of the RMSE 
of the global formation position error for each controller, evaluated in steady-state, considering the interval where the global formation error remains below $10^{-3}$ m. The results further confirm the superior precision of the AB--SMC.
\begin{table}[htbp]
\centering
\caption{RMSE of global formation position error from Monte Carlo simulations}
\begin{tabular}{lcc}
\toprule
\textbf{Controller} & \textbf{Median RMSE [m]} & \textbf{Std. deviation [m]} \\
\midrule
PBC    & $6.87\times10^{-4}$ & $\pm1.11\times10^{-7}$ \\
SMC    & $6.37\times10^{-4}$ & $\pm1.08\times10^{-5}$ \\
B--SMC & $1.72\times10^{-4}$ & $\pm6.87\times10^{-8}$ \\
AB--SMC & $1.72\times10^{-4}$ & $\pm7.96\times10^{-8}$ \\
\bottomrule
\end{tabular}
\label{tab:MC_RMSE}
\end{table}

Overall, the Monte Carlo analysis demonstrates the superior global control performance of the proposed AB--SMC, 
which achieves the best trade-off between transient response, steady-state accuracy, and control efficiency 
across a wide range of perturbed initial conditions.

\section{Conclusions}
An adaptive boundary sliding mode control (AB-SMC) strategy for ultra-precise spacecraft formation flying has been presented within a port-Hamiltonian framework. This formulation preserves the intrinsic nonlinear dynamics of the formation and enables a rigorous Lyapunov-based stability analysis. The trade-off inherent in conventional constant-gain boundary SMC between transient response and steady-state accuracy was addressed through an adaptive gain regulation mechanism based on the tracking error. Numerical simulations have shown that the proposed AB-SMC approach provides faster convergence, reduced control effort, and sub-millimeter tracking accuracy, while maintaining robustness to orbital perturbations and effectively suppressing chattering. The results indicate that the proposed method is a promosing and efficient solution for future space interferometry missions requiring high-precision and energy-efficient formation control.
%This work presented an AB-SMC strategy for ultra-precise spacecraft formation flying, developed within a port-Hamiltonian framework to preserve the intrinsic nonlinear dynamics of the formation and enable rigorous Lyapunov-based stability analysis. After identifying the trade-offs of constant-gain B-SMC between transient performance and steady-state accuracy, an adaptive mechanism is introduced to dynamically tune the control gains according to the tracking error. Numerical simulations demonstrate that AB-SMC achieves faster convergence, lower control effort, and submillimeter tracking precision, while maintaining robustness against orbital perturbations and effectively mitigating chattering. Overall, the proposed controller offers a promising solution for next-generation space interferometry missions requiring high-precision and energy-efficient formation control.

% \begin{ack}
% Place acknowledgments here.
% \end{ack}

% \section*{DECLARATION OF GENERATIVE AI AND AI-ASSISTED TECHNOLOGIES IN THE WRITING PROCESS}
% During the preparation of this work the author(s) used [NAME TOOL / SERVICE] in order to [REASON]. After using this tool/service, the author(s) reviewed and edited the content as needed and take(s) full responsibility for the content of the publication.
 \bibliographystyle{elsarticle-num} 
\bibliography{ifacconf}             % bib file to produce the bibliography
                                                     % with bibtex (preferred)

% \appendix
% \section{A summary of Latin grammar}    % Each appendix must have a short title.
% \section{Some Latin vocabulary}              % Sections and subsections are supported  
                                                                         % in the appendices.
\end{document}